\documentclass[final]{IEEEtran}
\usepackage{tikz}

\usepackage[T1]{fontenc}
\usepackage[utf8]{inputenc}
\usepackage{mathtools}

\usepackage{amssymb,mathrsfs}
\usepackage{amsthm}
\usepackage{bm}
\usepackage{scalerel}
\usepackage{nicefrac}
\usepackage{microtype} 
\usepackage[shortlabels]{enumitem}
\usepackage{graphicx}
\usepackage{epstopdf}
\DeclareGraphicsExtensions{.eps,.png,.jpg,.pdf}

\usepackage{url}
\usepackage{colortbl}
\usepackage{booktabs}
\usepackage{multirow}
\usepackage{colortbl,xcolor}
\usepackage{soul}
\usepackage{xparse,xstring}
\usepackage{calc}
\usepackage{etoolbox}

\makeatletter
\@ifpackageloaded{natbib}{
	\relax
}{
	\usepackage{cite}
}
\makeatother

\usepackage{array}
\newcolumntype{L}[1]{>{\raggedright\let\newline\\\arraybackslash\hspace{0pt}}m{#1}}
\newcolumntype{C}[1]{>{\centering\let\newline\\\arraybackslash\hspace{0pt}}m{#1}}
\newcolumntype{R}[1]{>{\raggedleft\let\newline\\\arraybackslash\hspace{0pt}}m{#1}}

\makeatletter
\let\MYcaption\@makecaption
\makeatother
\usepackage[font=footnotesize]{subcaption}
\makeatletter
\let\@makecaption\MYcaption
\makeatother

\usepackage{glossaries}
\makeatletter
\sfcode`\.1006

\let\oldgls\gls
\let\oldglspl\glspl

\newcommand\fussy@ifnextchar[3]{%
	\let\reserved@d=#1%
	\def\reserved@a{#2}%
	\def\reserved@b{#3}%
	\futurelet\@let@token\fussy@ifnch}
\def\fussy@ifnch{%
	\ifx\@let@token\reserved@d
		\let\reserved@c\reserved@a
	\else
		\let\reserved@c\reserved@b
	\fi
	\reserved@c}

\renewcommand{\gls}[1]{%
\oldgls{#1}\fussy@ifnextchar.{\@checkperiod}{\@}}
\renewcommand{\glspl}[1]{%
\oldglspl{#1}\fussy@ifnextchar.{\@checkperiod}{\@}}

\newcommand{\@checkperiod}[1]{%
	\ifnum\sfcode`\.=\spacefactor\else#1\fi
}

\robustify{\gls}
\robustify{\glspl}
\makeatother

\newacronym{wrt}{w.r.t.}{with respect to}
\newacronym{RHS}{R.H.S.}{right-hand side}
\newacronym{LHS}{L.H.S.}{left-hand side}
\newacronym{iid}{i.i.d.}{independent and identically distributed}
\newacronym{SOTA}{SOTA}{state-of-the-art}

\usepackage{float}

\ifx\notloadhyperref\undefined
	\ifx\loadbibentry\undefined
		\usepackage[hidelinks,hypertexnames=false]{hyperref}
	\else
		\usepackage{bibentry}
		\makeatletter\let\saved@bibitem\@bibitem\makeatother
		\usepackage[hidelinks,hypertexnames=false]{hyperref}
		\makeatletter\let\@bibitem\saved@bibitem\makeatother
	\fi
\else
	\ifx\loadbibentry\undefined
		\relax
	\else
		\usepackage{bibentry}
	\fi
\fi

\usepackage{cleveref-forward}

\crefname{equation}{}{}
\Crefname{equation}{}{}
\crefname{claim}{claim}{claims}
\crefname{step}{step}{steps}
\crefname{line}{line}{lines}
\crefname{condition}{condition}{conditions}
\crefname{dmath}{}{}
\crefname{dseries}{}{}
\crefname{dgroup}{}{}
\crefname{page}{page}{pages}

\crefname{Problem}{Problem}{Problems}
\crefformat{Problem}{Problem~#2#1#3}
\crefrangeformat{Problem}{Problems~#3#1#4 to~#5#2#6}

\crefname{Theorem}{Theorem}{Theorems}
\crefname{Corollary}{Corollary}{Corollaries}
\crefname{Proposition}{Proposition}{Propositions}
\crefname{Lemma}{Lemma}{Lemmas}
\crefname{Definition}{Definition}{Definitions}
\crefname{Example}{Example}{Examples}
\crefname{Assumption}{Assumption}{Assumptions}
\crefname{Remark}{Remark}{Remarks}
\crefname{Rem}{Remark}{Remarks}
\crefname{remarks}{Remarks}{Remarks}
\crefname{Appendix}{Appendix}{Appendices}
\crefname{Supplement}{Supplement}{Supplements}
\crefname{Exercise}{Exercise}{Exercises}
\crefname{TheoremA}{Theorem}{Theorems}
\crefname{CorollaryA}{Corollary}{Corollaries}
\crefname{PropositionA}{Proposition}{Propositions}
\crefname{LemmaA}{Lemma}{Lemmas}
\crefname{DefinitionA}{Definition}{Definitions}
\crefname{ExampleA}{Example}{Examples}
\crefname{RemarkA}{Remark}{Remarks}
\crefname{AssumptionA}{Assumption}{Assumptions}
\crefname{ExerciseA}{Exercise}{Exercises}
\crefname{algorithm}{Algorithm}{Algorithms}
\crefname{figure}{Fig.}{Figs.}
\crefname{table}{Table}{Tables}
\crefname{section}{Section}{Sections}
\crefname{subsection}{Section}{Sections}
\crefname{subsubsection}{Section}{Sections}

\usepackage{crossreftools}
\ifx\notloadhyperref\undefined
\else
	\relax
\fi

\usepackage{algorithm}
\usepackage{algpseudocode}

\ifx\loadbreqn\undefined
	\relax
\else
	\usepackage{breqn}
\fi

\makeatletter
\def\cleartheorem#1{%
    \expandafter\let\csname#1\endcsname\relax
    \expandafter\let\csname c@#1\endcsname\relax
}
\def\clearthms#1{ \@for\tname:=#1\do{\cleartheorem\tname} }
\makeatother

\ifx\renewtheorem\undefined
	\ifx\useTheoremCounter\undefined
		\newtheorem{Theorem}{Theorem}
		\newtheorem{Corollary}{Corollary}
		\newtheorem{Proposition}{Proposition}
		\newtheorem{Lemma}{Lemma}
	\else
		\newtheorem{Theorem}{Theorem}

	\fi

	\newtheorem{Definition}{Definition}
	
	\newtheorem{Remark}{Remark}

\fi

\theoremstyle{remark}

\theoremstyle{plain}

\newcommand{\qednew}{\nobreak \ifvmode \relax \else
		\ifdim\lastskip<1.5em \hskip-\lastskip
			\hskip1.5em plus0em minus0.5em \fi \nobreak
		\vrule height0.75em width0.5em depth0.25em\fi}



\newcommand{\nn}{\nonumber\\ }

\NewDocumentCommand{\movedownsub}{e{^_}}{%
	\IfNoValueTF{#1}{%
		\IfNoValueF{#2}{^{}}
	}{%
		^{#1}
	}%
	\IfNoValueF{#2}{_{#2}}
}

\let\latexchi\chi
\RenewDocumentCommand{\chi}{}{\latexchi\movedownsub}

\newcommand{\Real}{\mathbb{R}}

\newcommand{\Complex}{\mathbb{C}}

\newcommand{\calE}{\mathcal{E}}

\newcommand{\calG}{\mathcal{G}}
\newcommand{\calH}{\mathcal{H}}

\newcommand{\calM}{\mathcal{M}}

\newcommand{\calS}{\mathcal{S}}
\newcommand{\calT}{\mathcal{T}}
\newcommand{\calU}{\mathcal{U}}
\newcommand{\calV}{\mathcal{V}}
\newcommand{\calW}{\mathcal{W}}

\newcommand{\bA}{\mathbf{A}}

\newcommand{\bD}{\mathbf{D}}

\newcommand{\bh}{\mathbf{h}}

\newcommand{\bL}{\mathbf{L}}

\newcommand{\bt}{\mathbf{t}}

\newcommand{\bu}{\mathbf{u}}
\newcommand{\bU}{\mathbf{U}}

\newcommand{\bx}{\mathbf{x}}

\newcommand{\bY}{\mathbf{Y}}

\newcommand{\bbR}{\mathbb{R}}

\DeclareSymbolFont{bsfletters}{OT1}{cmss}{bx}{n}
\DeclareSymbolFont{ssfletters}{OT1}{cmss}{m}{n}
\DeclareMathSymbol{\bsfGamma}{0}{bsfletters}{'000}
\DeclareMathSymbol{\ssfGamma}{0}{ssfletters}{'000}
\DeclareMathSymbol{\bsfDelta}{0}{bsfletters}{'001}
\DeclareMathSymbol{\ssfDelta}{0}{ssfletters}{'001}
\DeclareMathSymbol{\bsfTheta}{0}{bsfletters}{'002}
\DeclareMathSymbol{\ssfTheta}{0}{ssfletters}{'002}
\DeclareMathSymbol{\bsfLambda}{0}{bsfletters}{'003}
\DeclareMathSymbol{\ssfLambda}{0}{ssfletters}{'003}
\DeclareMathSymbol{\bsfXi}{0}{bsfletters}{'004}
\DeclareMathSymbol{\ssfXi}{0}{ssfletters}{'004}
\DeclareMathSymbol{\bsfPi}{0}{bsfletters}{'005}
\DeclareMathSymbol{\ssfPi}{0}{ssfletters}{'005}
\DeclareMathSymbol{\bsfSigma}{0}{bsfletters}{'006}
\DeclareMathSymbol{\ssfSigma}{0}{ssfletters}{'006}
\DeclareMathSymbol{\bsfUpsilon}{0}{bsfletters}{'007}
\DeclareMathSymbol{\ssfUpsilon}{0}{ssfletters}{'007}
\DeclareMathSymbol{\bsfPhi}{0}{bsfletters}{'010}
\DeclareMathSymbol{\ssfPhi}{0}{ssfletters}{'010}
\DeclareMathSymbol{\bsfPsi}{0}{bsfletters}{'011}
\DeclareMathSymbol{\ssfPsi}{0}{ssfletters}{'011}
\DeclareMathSymbol{\bsfOmega}{0}{bsfletters}{'012}
\DeclareMathSymbol{\ssfOmega}{0}{ssfletters}{'012}

\newcommand{\bdelta}{\bm{\delta}}

\makeatletter
\newcommand*\rel@kern[1]{\kern#1\dimexpr\macc@kerna}
\newcommand*\widebar[1]{%
  \begingroup
  \def\mathaccent##1##2{%
    \rel@kern{0.8}%
    \overline{\rel@kern{-0.8}\macc@nucleus\rel@kern{0.2}}%
    \rel@kern{-0.2}%
  }%
  \macc@depth\@ne
  \let\math@bgroup\@empty \let\math@egroup\macc@set@skewchar
  \mathsurround\z@ \frozen@everymath{\mathgroup\macc@group\relax}%
  \macc@set@skewchar\relax
  \let\mathaccentV\macc@nested@a
  \macc@nested@a\relax111{#1}%
  \endgroup
}
\makeatother

\DeclareMathOperator*{\argmax}{arg\,max}
\DeclareMathOperator*{\argmin}{arg\,min}

\DeclareMathOperator{\ST}{s.t.\ }

\DeclareMathOperator{\diag}{diag}

\DeclareMathOperator{\var}{var}

\DeclareMathOperator{\cov}{cov}

\DeclareMathOperator{\rank}{rank}

\DeclareMathOperator{\ima}{im}

\newcommand{\ifbcdot}[1]{\ifblank{#1}{\cdot}{#1}}

\DeclarePairedDelimiterX\abs[1]{\lvert}{\rvert}{\ifbcdot{#1}}
\DeclarePairedDelimiterX\parens[1]{(}{)}{\ifbcdot{#1}}
\DeclarePairedDelimiterX\brk[1]{[}{]}{\ifbcdot{#1}}
\DeclarePairedDelimiterX\braces[1]{\{}{\}}{\ifbcdot{#1}}
\DeclarePairedDelimiterX\angles[1]{\langle}{\rangle}{\ifblank{#1}{\cdot,\cdot}{#1}}
\DeclarePairedDelimiterX\ip[2]{\langle}{\rangle}{\ifbcdot{#1},\ifbcdot{#2}}
\DeclarePairedDelimiterX\norm[1]{\lVert}{\rVert}{\ifbcdot{#1}}
\DeclarePairedDelimiterX\ceil[1]{\lceil}{\rceil}{\ifbcdot{#1}}
\DeclarePairedDelimiterX\floor[1]{\lfloor}{\rfloor}{\ifbcdot{#1}}

\DeclareFontFamily{U}{matha}{\hyphenchar\font45}
\DeclareFontShape{U}{matha}{m}{n}{
      <5> <6> <7> <8> <9> <10> gen * matha
      <10.95> matha10 <12> <14.4> <17.28> <20.74> <24.88> matha12
      }{}
\DeclareSymbolFont{matha}{U}{matha}{m}{n}
\DeclareFontSubstitution{U}{matha}{m}{n}

\DeclareFontFamily{U}{mathx}{\hyphenchar\font45}
\DeclareFontShape{U}{mathx}{m}{n}{
      <5> <6> <7> <8> <9> <10>
      <10.95> <12> <14.4> <17.28> <20.74> <24.88>
      mathx10
      }{}
\DeclareSymbolFont{mathx}{U}{mathx}{m}{n}
\DeclareFontSubstitution{U}{mathx}{m}{n}

\DeclareMathDelimiter{\vvvert}{0}{matha}{"7E}{mathx}{"17}
\DeclarePairedDelimiterX\vertiii[1]{\vvvert}{\vvvert}{\ifbcdot{#1}}

\DeclarePairedDelimiterXPP\trace[1]{\operatorname{Tr}}{(}{)}{}{\ifbcdot{#1}} 
\DeclarePairedDelimiterXPP\col[1]{\operatorname{col}}{\{}{\}}{}{\ifbcdot{#1}} 
\DeclarePairedDelimiterXPP\row[1]{\operatorname{row}}{\{}{\}}{}{\ifbcdot{#1}} 
\DeclarePairedDelimiterXPP\erf[1]{\operatorname{erf}}{(}{)}{}{\ifbcdot{#1}}
\DeclarePairedDelimiterXPP\erfc[1]{\operatorname{erfc}}{(}{)}{}{\ifbcdot{#1}}
\DeclarePairedDelimiterXPP\KLD[2]{D}{(}{)}{}{\ifbcdot{#1}\, \delimsize\|\, \ifbcdot{#2}} 
\DeclarePairedDelimiterXPP\op[2]{\operatorname{#1}}{(}{)}{}{#2} 

\DeclarePairedDelimiterXPP\indicate[1]{{\bf 1}}{\{}{\}}{}{\ifbcdot{#1}}

\newcommand{\tc}[1]{^{(#1)}}
\NewDocumentCommand\ofrac{s m}{%
	\IfBooleanTF#1%
	{\dfrac{1}{#2}}%
	{\frac{1}{#2}}%
}
\NewDocumentCommand\ddfrac{s m m}{%
	\IfBooleanTF#1%
	{\dfrac{\mathrm{d} {#2}}{\mathrm{d} {#3}}}%
	{\frac{\mathrm{d} {#2}}{\mathrm{d} {#3}}}%
}
\NewDocumentCommand\ppfrac{s m m}{%
	\IfBooleanTF#1%
	{\dfrac{\partial {#2}}{\partial {#3}}}%
	{\frac{\partial {#2}}{\partial {#3}}}%
}

\newcommand{\setgiven}{:}
\providecommand\given{}

\DeclarePairedDelimiterX\Set[2]\{\}{%
	\if#1:%
		\renewcommand\given{\SetSymbol{:}}%
	\else%
		\renewcommand\given{\SetSymbol[\delimsize]{#1}}%
	\fi%
#2
}

\NewDocumentCommand\set{s O{\setgiven} m}{%
	\IfBooleanTF#1%
	{\Set*{#2}{#3}}%
	{\Set{#2}{#3}}%
}

\NewDocumentCommand{\evalat}{ s O{\big} m e{_^} }{%
\IfBooleanTF{#1}%
{\left. #3 \right|}{#3#2|}%
\IfValueT{#4}{_{#4}}%
\IfValueT{#5}{^{#5}}%
}

\providecommand\given{}
\DeclarePairedDelimiterXPP\cprob[1]{}(){}{
\renewcommand\given{\nonscript\,\delimsize\vert\allowbreak\nonscript\,\mathopen{}}%
#1%
}
\DeclarePairedDelimiterXPP\cexp[1]{}[]{}{
\renewcommand\given{\nonscript\,\delimsize\vert\allowbreak\nonscript\,\mathopen{}}%
#1%
}

\DeclareDocumentCommand \P { s e{_^} d() g } {%
	\mathbb{P}%
	\IfBooleanTF{#1}%
		{
			\IfValueT{#2}{_{#2}}%
			\IfValueT{#3}{^{#3}}%
			\IfValueTF{#5}{\cprob{#4 \given #5}}{\IfValueT{#4}{\cprob{#4}}}%
		}%
		{
			\IfValueT{#2}{_{#2}}%
			\IfValueT{#3}{^{#3}}%
			\IfValueTF{#5}{\cprob*{#4 \given #5}}{\IfValueT{#4}{\cprob*{#4}}}%
		}%
}

\DeclareDocumentCommand \E { s e{_^} o g } {%
	\mathbb{E}%
	\IfBooleanTF{#1}%
		{
			\IfValueT{#2}{_{#2}}%
			\IfValueT{#3}{^{#3}}%
			\IfValueTF{#5}{\cexp{#4 \given #5}}{\IfValueT{#4}{\cexp{#4}}}%
		}%
		{
			\IfValueT{#2}{_{#2}}%
			\IfValueT{#3}{^{#3}}%
			\IfValueTF{#5}{\cexp*{#4 \given #5}}{\IfValueT{#4}{\cexp*{#4}}}%
		}%
}

\DeclareDocumentCommand \Var { s e{_^} d() g } {%
	\var%
	\IfBooleanTF{#1}%
		{
			\IfValueT{#2}{_{#2}}%
			\IfValueT{#3}{^{#3}}%
			\IfValueTF{#5}{\cprob{#4 \given #5}}{\IfValueT{#4}{\cprob{#4}}}%
		}%
		{
			\IfValueT{#2}{_{#2}}%
			\IfValueT{#3}{^{#3}}%
			\IfValueTF{#5}{\cprob*{#4 \given #5}}{\IfValueT{#4}{\cprob*{#4}}}%
		}%
}

\DeclareDocumentCommand \Cov { s e{_^} d() g } {%
	\cov%
	\IfBooleanTF{#1}%
		{
			\IfValueT{#2}{_{#2}}%
			\IfValueT{#3}{^{#3}}%
			\IfValueTF{#5}{\cprob{#4 \given #5}}{\IfValueT{#4}{\cprob{#4}}}%
		}%
		{
			\IfValueT{#2}{_{#2}}%
			\IfValueT{#3}{^{#3}}%
			\IfValueTF{#5}{\cprob*{#4 \given #5}}{\IfValueT{#4}{\cprob*{#4}}}%
		}%
}

\ExplSyntaxOn
\NewDocumentCommand \dist {m o o} {%
\mathrm{#1}\left(%
	\IfValueT{#3}{%
		\tl_if_blank:nTF{ #3 }{\cdot\, \middle|\, }{#3\, \middle|\, }%
	}
	\IfValueT{#2}{#2}%
\right)%
}
\ExplSyntaxOff

\NewDocumentCommand {\cbrace} {t+ D[]{black} D(){\widthof{#5}} m m } {%
	\begingroup%
		\color{#2}
		\IfBooleanTF{#1}{%
			\overbrace{#4}^%
		}{
			\underbrace{#4}_%
		}%
		{\parbox[c]{#3}{\centering\footnotesize{#5}}}%
	\endgroup%
}

\let\oldforall\forall
\renewcommand{\forall}{\oldforall \, }

\let\oldexist\exists
\renewcommand{\exists}{\oldexist \, }

\makeatletter

\newcommand{\rankcolor}[2]{%
	\expandafter\renewcommand\csname #1\endcsname[1]{%
		\ifblank{##1}{%
			{\color{#2} \textbf{#2}}%
		}{%
			\ifmmode
				\textcolor{#2}{\bm{##1}}%
			\else%
				{\color{#2} \textbf{##1}}%
			\fi	
		}%
	}
}

\rankcolor{first}{red}
\rankcolor{second}{blue}
\rankcolor{third}{cyan}
\makeatother

\graphicspath{{./Figures/}{./figures/}}
\DeclareDocumentCommand{\includeCroppedPdf}{ o O{./Figures/} m }{
	\IfFileExists{#2#3-crop.pdf}{}{%
		\immediate\write18{pdfcrop #2#3.pdf #2#3-crop.pdf}}%
	\includegraphics[#1]{#2#3-crop.pdf}
}

\makeatletter
\newcommand*{\addFileDependency}[1]{
  \typeout{(#1)}
  \@addtofilelist{#1}
  \IfFileExists{#1}{}{\typeout{No file #1.}}
}
\makeatother

\definecolor{gray90}{gray}{0.9}
\def\colorlist{red,blue,brown,cyan,darkgray,gray,lightgray,green,lime,magenta,olive,orange,pink,purple,teal,violet,white,yellow}

\makeatletter
\def\startcomment{[}
\ifx\nohighlights\undefined
	\newcommand{\createcolor}[1]{%
			\expandafter\newcommand\csname #1\endcsname[1]{{\color{#1} ##1}}%
	}
	\newcommand{\msout}[1]{\text{\color{green} \st{\ensuremath{#1}}}}
	\newcommand{\del}[1]{{\color{green}\ifmmode \msout{#1}\else\st{#1}\fi}}
\else
	\newcommand{\createcolor}[1]{%
			\expandafter\newcommand\csname #1\endcsname[1]{%
				\noexpandarg%
				\StrChar{##1}{1}[\firstletter]%
				\if\firstletter\startcomment%
					\relax
				\else%
					##1
				\fi
			}%
	}
	\newcommand{\msout}[1]{}
	\newcommand{\del}[1]{}
\fi

\def\@tempa#1,{%
    \ifx\relax#1\relax\else
        \createcolor{#1}%
        \expandafter\@tempa
    \fi
}
\expandafter\@tempa\colorlist,\relax,
\makeatother

\newcommand{\hhide}[1]{}

\ifx\diagnoselabel\undefined
	\relax
\else
	\makeatletter
	\def\@testdef #1#2#3{%
		\def\reserved@a{#3}\expandafter \ifx \csname #1@#2\endcsname
			\reserved@a  \else
			\typeout{^^Jlabel #2 changed:^^J%
				\meaning\reserved@a^^J%
				\expandafter\meaning\csname #1@#2\endcsname^^J}%
			\@tempswatrue \fi}
	\makeatother
\fi

\newcommand{\boldell}{{\bm{\ell}}}

\newacronym{GSP}{GSP}{graph signal processing}
\newacronym{GGSP}{GGSP}{generalized graph signal processing}
\newacronym{PSWF}{PSWF}{prolate spheroidal wave function}
\newacronym{GFT}{GFT}{graph Fourier transform}
\newacronym{WLOG}{WLOG}{without loss of generality}
\newacronym{JFT}{JFT}{joint Fourier transform}
\newacronym{STVFT}{STVFT}{short time-vertex Fourier transform}
\newacronym{STVWT}{STVWT}{spectral time-vertex wavelet transform}

\begin{document}
\def\BibTeX{{\rm B\kern-.05em{\sc i\kern-.025em b}\kern-.08em
    T\kern-.1667em\lower.7ex\hbox{E}\kern-.125emX}}
\title{Generalized Graph Signal Reconstruction via the Uncertainty Principle}
\author{Yanan Zhao$^{1}$, Xingchao Jian$^{1}$, Feng Ji$^{1}$, Wee Peng Tay$^{1}$ and Antonio Ortega$^{2}$\\
{$^{1}$School of Electrical and Electronic Engineering, Nanyang Technological University, Singapore,\\
$^{2}$Department of Electrical and Computer Engineering,  University of Southern California, Los Angeles, USA}
}

%
%
%

%
\maketitle
\begin{abstract}
We introduce a novel uncertainty principle for generalized graph signals that extends classical time-frequency and graph uncertainty principles into a unified framework. By defining joint vertex-time and spectral-frequency spreads, we quantify signal localization across these domains, revealing a trade-off between them. This framework allows us to identify a class of signals with maximal energy concentration in both domains, forming the fundamental atoms for a new joint vertex-time dictionary. This dictionary enhances signal reconstruction under practical constraints, such as incomplete or intermittent data, commonly encountered in sensor and social networks. Numerical experiments on real-world datasets demonstrate the effectiveness of the proposed approach, showing improved reconstruction accuracy and noise robustness compared to existing methods.
\end{abstract}
\begin{IEEEkeywords}
Uncertainty principle, localized representation, joint vertex-time dictionary, generalized graph signal reconstruction
\end{IEEEkeywords}
\section{Introduction}
\label{sec:intro}
The uncertainty principle, developed by Slepian, Landau, and Pollack \cite{Slepian1983,Laudau1980}, is fundamental in signal processing, emphasizing the trade-off between a signal's time duration and frequency bandwidth, indicating that a signal cannot be precisely localized in both domains simultaneously. \Glspl{PSWF} \cite{Slepian1961,Pollak1961,Laudau1962,Slepian1978,orfanidis1996} optimally balance time and frequency localization by maximizing energy concentration within a finite time interval while being band-limited to a specific frequency range. These functions serve as efficient basis functions for signal representation, minimizing information loss under constraints in both time and frequency domains. Analogously, an uncertainty principle for signals on graphs \cite{Tsitsvero2016} describes a trade-off between localization in the graph and spectral domains with vertex and spectral spreads reflecting the concentration of the signal's energy in specific vertex subsets and spectral components, respectively. A class of graph signals that are maximally concentrated in both domains forms the basis for constructing dictionaries \cite{Tsitsvero2015} that enhance robust graph signal representation, particularly when dealing with missing or damaged vertices.

Traditional \gls{GSP} primarily focuses on information in the graph domain rather than the time domain. The vertex-time GSP theory \cite{Loukas2016,Perraudin2015,Zhang2021} addresses this gap by integrating information from both the graph and discrete-time domain. It introduced joint vertex-time Fourier transforms, filters, harmonic analysis, and redundant joint vertex-time dictionaries \cite{Grassi2016,Grassi2018} for localized signal representations, such as the short vertex-time Fourier transform (STVFT) and spectral time-vertex wavelet transform (STVWT). Motivated by applications in sensor and social networks, where nodes observe time-varying continuous signals, this theory has been further extended to accommodate vertex signals from a potentially infinite-dimensional, separable Hilbert space \cite{JiTay:J19}, termed the \gls{GGSP} framework.

Practical challenges, including sensor damage, limited operation times, and environmental factors, can restrict data availability and complicate signal reconstruction. Robust signal processing techniques are thus required to manage incomplete or intermittent data, ensuring reliable performance in dynamic and constrained settings. To address these challenges, it is essential to design basis functions that are maximally energy-concentrated over subsets of vertices and time intervals in the joint vertex-time domain. These functions effectively capture available information during operational periods by focusing on localized subsets relevant to the sensor data. Inspired by energy-concentrated basis functions derived from uncertainty principles for graph and continuous-time signals, we propose an uncertainty principle for generalized graph signals and utilizing it to design energy-concentrated basis functions for signal reconstruction.

The contributions of this paper are threefold: a) We introduce a novel uncertainty principle for generalized graph signals, defining joint vertex-time and spectral-frequency spreads to measure a signal’s localization in these domains. This principle includes existing uncertainty principles for both graph and continuous-time signals as special cases. b) We identify a class of generalized graph signals that are maximally concentrated in both joint vertex-time and spectral-frequency domains. We propose to use these signals as the dictionary for signal reconstruction. We derive their localization properties, establish conditions for perfect localization in both domains, and outline the construction and optimization of the dictionary to improve reconstruction accuracy. c) We test our proposed joint vertex-time dictionary learning algorithm's performance on real datasets, demonstrating its effectiveness with limited data and its superior noise resilience compared to other methods.

\section{Basic definitions}
In this section, we introduce the basic definitions in \gls{GGSP} that are used in this paper. Consider a simple, connected, undirected graph $\calG=(\calV,\calE)$ with $N$ vertices $\calV=\set{1, 2, \dots, N}$ and weighted edges $\mathcal{E}=\set{a_{ij}}_{i,j\in \mathcal{V}}$. The graph has no self-loops, and the edge weights satisfy $a_{ij} > 0$ if vertices $i$ and $j$ are connected, and $a_{ij} = 0$ otherwise. The adjacency matrix $\bA = [a_{ij}]$ is symmetric with zero diagonal entries. The degree of vertex $i$ is $d_{i}:=\sum_{j=1}^{N}a_{ij}$, and the degree matrix is $\bD=\diag \set{d_1,d_2,\dots,d_N}$. The graph Laplacian matrix is given by $\bL=\bD-\bA$. We take $\bL$ to be the graph shift operator (GSO). In \gls{GSP}, a finite-energy graph signal can be written as a function in $L^2(\calV)$ or a vector in $\Complex^N$. The Laplacian matrix $\bL$ has an eigen-decomposition $ \bL = \bU\Lambda\bU^{\ast}=\sum_{i=1}^{N}\lambda_{i}\bu_{i}\bu_{i}^{\ast}$, where $\Lambda$ is a diagonal matrix with non-negative real eigenvalues $\set{\lambda_{i}}_{i=1,\dots,N}$, and real-valued orthonormal eigenvectors $\set{\bu_{i}}_{i=1,\dots,N}$. Here, $(\cdot)^{\ast}$ denotes conjugate transpose. The graph Fourier transform (GFT) of a graph signal $\bx$ defined over an undirected graph has been defined in \cite{Shuman2013,Pesenson2010,Sandryhaila2013,Pesenson2008,Rabbat2012,ortega2018graph} as $\widehat{\bx}=\bU^{\ast}\bx$
where $\bU$ is the unitary matrix whose columns are the Laplacian eigenvectors.
In \gls{GGSP}, a generalized graph signal $f$ is defined as a map from the vertex set $\calV$ to a separable Hilbert space $\calH$. In this work, we consider $\calH=L^{2}(\calT)$ (with the usual Lebesgue measure), where $\calT=\bbR$ represents the time domain.
Therefore, a generalized graph signal $f$ can be identified with the map $\calV\times \calT \rightarrow \Real$, defined by $(v,t)\mapsto f(v, t)$. In this paper, the space of generalized graph signals is denoted by $L^{2}\left(\calV\times\calT\right)$, where $\calV$ and $\calT$ are referred as the \emph{vertex} and \emph{time} domains, respectively. For $f\in L^2(\calV\times \calT)$, the joint Fourier transform is defined as \cite{JiTay:J19}:
\begin{align}
\label{eq.joint_Fourier_transform}
\mathcal{F}_{f}(\lambda_k,\omega)= \langle f, \bu_{k}\otimes e^{j\omega t}\rangle,
\end{align}
where $\lambda_k$ and $\bu_{k}$ are the $k$-th eigenvalue and eigenvector of $\bL$, respectively, and $\otimes$ is the tensor product. This transform maps the signal $f$ in $L^2(\calV\times\calT)$ having the joint vertex-time domain $\calV\times\calT$ to a signal in $L^2(\calW\times\Omega)$ with the joint spectral-frequency domain $\calW\times\Omega$, where $\calW = \set{\lambda_1,\dots,\lambda_N}$ and $\Omega=\bbR$ are referred to as the \emph{spectral} and \emph{frequency} domains, respectively. It generalizes both \gls{GFT} and the classical Fourier transform (FT). Specifically, it reduces to the \gls{GFT}, which maps a graph signal from $\calV$ to $\calW$, when $f$ is evaluated at a fixed time instance, and to the FT mapping a temperal signal from $\calT$ to $\Omega$ when $\calV$ is a singleton.
\section{Uncentainty Principle for Generalized Graph Signals}
\label{sec:localized properties}
In this section, we first define the notions of joint vertex-time and spectral-frequency spreads, which quantify signal localization in $\calV\times\calT$ and $\calW\times\Omega$. We then derive the uncertainty principle for generalized graph signals, demonstrating the trade-off between these spreads. This principle extends the existing uncertainty principles for graph and temporal signals as special cases.

\subsection{Joint vertex-time and spectral-frequency spread}
We first define the spread in $\calV\times\calT$ and $\calW\times\Omega$ domains, measuring signals' localization in these domains, respectively.

\begin{Definition}\label{def:vt-limiting}
(Joint vertex-time limiting operator) Given a subset $\calS \subseteq \calV\times \calT$ and a generalized graph signal $f\in L^{2}(\calV\times \calT)$, a joint vertex-time limiting operator satisfies
\begin{align}\label{eq.joint_VT_opt}
\Pi_{\calS} f:=
\begin{cases}
f(v,t) &\mathrm{if}~(v,t)\in~\calS,\\
0  &\mathrm{otherwise}.
\end{cases}
\end{align}
$\Pi_{\calS}$ is a Hermitian operator, and satisfies $\Pi_{\calS}^2=\Pi_{\calS}$.
\end{Definition}

\begin{Remark}
The joint vertex-time limiting operator $\Pi_{\calS}$ generalizes both the vertex-limiting operator in the graph domain and the time-limiting operator in time domain:
\begin{enumerate}[i)]
\item When $\calH=\Complex$, $f\in L^{2}(\calV\times\calT)$ reduces to a graph signal $f(v)\in \Complex$ for each $v\in \calV$. In this case, $\Pi_{\calS}$ reduces to the vertex-limiting operator $\Pi_{\calV'}$ \cite{Tsitsvero2016}, as defined by a diagonal matrix $\bD_{\calV'}=\diag\set{\mathbf{1}_{\calV'}}$ for a vertex subset $\calV'\subseteq\calV$.
\item For a trivial graph $\calG$ with a single vertex, $f\in L^{2}(\calV\times \calT)$ becomes a continuous-time signal $f(t)$ over $\calT$. Let $\calT'$ be an interval within $\calT$. The operator $\Pi_{\calS}$ reduces to the time-limiting operator $\Pi_{\calT'}$ \cite{Slepian1961}, which restricts the signal $f(t)$ to $\calT' = [t_{c} - \ell/2, t_{c} + \ell/2]$, where $t_{c}$ is the center time point and $\ell$ is the length of $\calT'$.
\end{enumerate}
\end{Remark}

\begin{Definition}\label{def:sf-limiting}
(Joint spectral-frequency limiting operator) Given a unitary transform $\calU: L^{2}(\calV\times \calT) \rightarrow L^{2}(\calW\times \Omega)$ and a subset $\Sigma \subseteq \calW \times \Omega$, the joint spectral-frequency limiting operator $\Pi_{\Sigma}$ for a generalized graph signal $f\in L^2(\calV\times\calT)$ is defined as $\Pi_{\Sigma} f = \calU^{-1}\Sigma\{\calU f\}$,
where\footnote{We omit $\calU$ in the notation $\Pi_\Sigma$ to avoid clutter as it will be clear from the context.}
\begin{align}
\Sigma\{\calU f\}:
= \begin{cases}
\left(\calU f\right)(w,\omega) &\mathrm{if}~(w,\omega)\in~\Sigma,\\
0  &\mathrm{otherwise}.
\end{cases}
\end{align}
$\Pi_{\Sigma}$ is a Hermitian operator, satisfying $\Pi_{\Sigma}^2=\Pi_{\Sigma}$.
\end{Definition}

A typical unitary transform $\calU$ mapping from $L^{2}(\calV\times \calT)$ to $L^{2}(\calW\times \Omega)$ is the \gls{JFT}, as defined in \cref{eq.joint_Fourier_transform}.
The joint spectral-frequency limiting operator $\Pi_{\Sigma}$ includes the following special cases:
\begin{enumerate} [i)]
    \item When $\calH=\Complex$, $f\in L^{2}(\calV\times\calT)$ becomes a graph signal. Then $\calU$ is \gls{GFT}, and $\Pi_{\Sigma}$ reduces to the spectral-limiting $\Pi_{\calW'}=\bU\Sigma_{\calW'}\bU^{\ast}$ with $\Sigma_{\calW'}=\diag\{\bf{1}_{\calW'}\}$, projecting $f$ onto the subspace spanned by the eigenvectors indexed by $\calW'$ \cite{Tsitsvero2016}.
    \item  For a trivial graph $\calG$ with a single vertex, $f\in L^{2}(\calV\times\calT)$ is a continuous-time function $f(t)$. In this case, $\calU$ becomes FT, and $\Pi_{\Sigma}$ reduces to the frequency-limiting operator $\Pi_{\Omega'}$ \cite{Slepian1961} which yields a bandlimited function $ \Pi_{\Omega'}f$ within the frequency interval $\Omega'= [\omega_{c}-W, \omega_{c}+W]$, where $\omega_{c}$ is the center frequency and $W$ denotes the bandwidth.

\end{enumerate}

The joint vertex-time limiting operator $\Pi_{\calS}$ and the joint spectral-frequency limiting operator $\Pi_{\Sigma}$ are orthogonal projections that map a generalized graph signal $f$ into the joint vertex-time limiting and joint spectral-frequency limiting spaces, respectively. The image of $\Pi_{\calS}$, denoted as $\ima(\Pi_{\calS})$, consists of signals satisfying
\begin{align}
\label{eq.vertex_time_limited_signals}
    \Pi_{\calS}f=f,
\end{align}
indicating perfect localization over the subset $\calS$. Similarly, $\ima(\Pi_{\Sigma})$ contains signals for which
\begin{align}
\label{eq.spec_fre_limited_signals}
    \Pi_{\Sigma}f=f,
\end{align}
signifying perfect localization over the spectral-frequency subset $\Sigma$, referred to as bandlimited signals.
For a subset $\calS$, its complement $\bar\calS$ is defined such that $\calS\cup\bar\calS=\calV\times\calT$ and $\calS\cap\bar\calS=\emptyset$. The projection onto $\bar\calS$ is denoted by $\overline{\Pi}_{\calS}$, and the projection onto the complementary set $\bar\Sigma$ is denoted by $\overline{\Pi}_{\Sigma}$. Both $\overline{\Pi}_{\calS}$ and $\overline{\Pi}_{\Sigma}$ are Hermitian, satisfying $\overline{\Pi}_{\calS}^2=\overline{\Pi}_{\calS}$ and $\overline{\Pi}_{\Sigma}^2=\overline{\Pi}_{\Sigma}$. Additionally, the identities $\Pi_{\calS} f+\overline{\Pi}_{\calS} f=f$ and $\Pi_{\Sigma} f+\overline{\Pi}_{\Sigma} f=f$ hold.

\begin{Definition}
(Joint vertex-time spread and joint spectral-frequency spread) For a generalized graph signal $f\in L^{2}(\calV\times \calT)$, the joint vertex-time limiting operator $\Pi_{\calS}$ and the joint spectral-frequency limiting operator $\Pi_{\Sigma}$ project $f$ onto the subsets $\calS$ and $\Sigma$, respectively. The joint vertex-time spread $\alpha_{\calS}^2$ and joint spectral-frequency spread $ \beta^2_{\Sigma}$ are defined as
\begin{align}
    \alpha^2_{\calS} = \frac{\norm{\Pi_{\calS} f}_{2}^{2}}{\norm{f}_{2}^{2}}=\frac{\norm{f}^2_{L^{2}(\calS)}}{\norm{f}^2_{L^{2}(\cal{V}\times \cal{T})}}
\label{eq.joint_vT_spread}
\end{align}
and
\begin{align}
    \beta^2_{\Sigma} = \frac{\norm{\Pi_{\Sigma} f}_{2}^{2}}{\norm{f}_{2}^{2}}=\frac{\norm{\Sigma\{\calU f\}}_{2}^{2}}{\norm{\calU f}_{2}^{2}}=\frac{\norm{\calU f}_{L^{2}(\Sigma)}^{2}}{\norm{\calU f}^{2}_{L^{2}(\calW\times \Omega)}}.
\label{eq.joint_SpFr_spread}
\end{align}
These spreads measure the energy concentration of $f$ within the subsets $\calS$ and $\Sigma$.
\end{Definition}

\subsection{Uncertainty principle}
\label{ssec:subhead}
Given a joint vertex-time subset $\calS$ and a joint spectral-frequency subset $\Sigma$, with corresponding energy concentrations $\alpha^2_{\calS}$ and $\beta^2_{\Sigma}$, an uncertainty principle for generalized graph signals characterizes the trade-off between $\alpha_{\calS}$ and $\beta_{\Sigma}$ and identifies the signals that can achieve all admissible pairs. This trade-off is presented in the following theorem.

\begin{Theorem}
\label{thm:UP_GGSP}
The feasible region of all possible pairs of $(\alpha_{\calS}, \beta_{\Sigma})$ is given by
\begin{align}
\label{eq.Theorem_feasible region}
\begin{aligned}
&\Theta = \Bigg\{ (\alpha_{\calS},\beta_{\Sigma}):\\
&\cos^{-1} \alpha_{\calS} + \cos^{-1} \beta_{\Sigma} \geq \cos^{-1}  \sqrt{\lambda_{\max}(\Pi_{\Sigma}\Pi_{\calS}\Pi_{\Sigma})},\\
&\cos^{-1}  \sqrt{1 - \alpha_{\calS}^2}+\cos^{-1} \beta_{\Sigma} \geq \cos^{-1}  \sqrt{\lambda_{\max}(\Pi_{\Sigma}\overline{\Pi}_{\calS}\Pi_{\Sigma})},\\
&\cos^{-1}  \alpha_{\calS} + \cos^{-1}  \sqrt{1 - \beta_{\Sigma}^2} \geq \cos^{-1}  \sqrt{\lambda_{\max}(\overline{\Pi}_{\Sigma}\Pi_{\calS}\overline{\Pi}_{\Sigma})}\\
&\cos^{-1}  \sqrt{1 - \alpha_{\calS}^2} + \cos^{-1}  \sqrt{1 - \beta_{\Sigma}^2} \\
& \qquad\qquad \geq \cos^{-1}  \sqrt{\lambda_{\max}(\overline{\Pi}_{\Sigma}\overline{\Pi}_{\calS}\overline{\Pi}_{\Sigma})}.
\Bigg\}
\end{aligned}
\end{align}
where $\lambda_{\max}(\cdot)$ denotes the largest eigenvalue of its operator argument.
\end{Theorem}
\begin{proof}
See \cref{prf:thm:UP_GGSP}.
\end{proof}

\begin{figure}
\centering
\begin{tikzpicture}[scale=4.3]

    \draw[dashed, dash pattern=on 2pt off 1pt] (0, 1) -- (1, 1);
    \draw[dashed, dash pattern=on 2pt off 1pt] (1, 0) -- (1, 1);

    \draw[->] (0, 0) -- (1.1, 0) node[right] {\scriptsize $\alpha_{\calS}^2$};
    \draw[->] (0, 0) -- (0, 1.1) node[above] {\scriptsize $\beta_{\Sigma}^2$};

    \shade[shading=axis, left color=gray!5, right color=gray!150, shading angle=135]
    (0, 0.8) to[out=80, in=180] (0.25, 1) --
    (0.75, 1) to[out=-10, in=100] (1, 0.75) --
    (1, 0.25) to[out=-80, in=-5] (0.8, 0) --
    (0.2, 0) to[out=180, in=-80] (0, 0.2) --
    (0, 0.8);

    \draw[color={rgb,255:red,0; green,100; blue,0}, thick, smooth] (0, 0.8) to[out=80, in=180] (0.25, 1);
    \draw[color={rgb,255:red,0; green,100; blue,0}, thick, smooth] (0.75, 1) to[out=-10, in=100] (1, 0.75);
    \draw[color={rgb,255:red,0; green,100; blue,0}, thick, smooth] (0, 0.2) to[out=-80, in=180] (0.2, 0);
    \draw[color={rgb,255:red,0; green,100; blue,0}, thick, smooth] (0.8, 0) to[out=-5, in=-80] (1, 0.25);
    \draw[color={rgb,255:red,0; green,100; blue,0}, thick] (0.25, 1) -- (0.75, 1);
    \draw[color={rgb,255:red,0; green,100; blue,0}, thick] (0, 0.2) -- (0, 0.8);
    \draw[color={rgb,255:red,0; green,100; blue,0}, thick] (0.2, 0) -- (0.8, 0);
    \draw[color={rgb,255:red,0; green,100; blue,0}, thick] (1, 0.25) -- (1, 0.75);

    \filldraw[fill=gray, draw=black] (0, 0.8) circle (0.3pt) node[left,font=\tiny] {$\lambda_{\max} (\overline{\Pi}_{ \calS} \Pi_{ \Sigma} \overline{\Pi}_{ \calS} )$};
    \filldraw[fill=gray, draw=black] (0, 0.2) circle (0.3pt) node[left,font=\tiny] {$1  -  \lambda_{\max} (\overline{\Pi}_{ \calS} \overline{\Pi}_{ \Sigma} \overline{\Pi}_{ \calS})$};
    \filldraw[fill=gray, draw=black] (0.2, 0) circle (0.3pt) node[below,font=\tiny] {$1  -  \lambda_{\max} (\overline{\Pi}_{ \Sigma} \overline{\Pi}_{ \calS} \overline{\Pi}_{ \Sigma})$};
    \filldraw[fill=gray, draw=black] (0.8, 0) circle (0.3pt) node[below left,xshift=2.1em,font=\tiny] {$\lambda_{\max} (\overline{\Pi}_{ \Sigma} \Pi_{ \calS} \overline{\Pi}_{ \Sigma} )$};\
    \filldraw[fill=gray, draw=black] (0.25, 1) circle (0.3pt) node[above right,xshift=-2.2em,font=\tiny] {$1  -  \lambda_{\max} (\Pi_{ \Sigma} \overline{\Pi}_{ \calS} \Pi_{ \Sigma})$};
    \filldraw[fill=gray, draw=black] (0.75, 1) circle (0.3pt) node[above right,xshift=-1.1em, font=\tiny] {$\lambda_{\max} (\Pi_{ \Sigma} \Pi_{ \calS} \Pi_{ \Sigma})$};
    \filldraw[fill=gray, draw=black] (1, 0.75) circle (0.3pt) node[right,font=\tiny] {$\lambda_{\max} (\Pi_{ \calS} \Pi_{ \Sigma} \Pi_{ \calS} )$};
    \filldraw[fill=gray, draw=black] (1, 0.25) circle (0.3pt) node[right,font=\tiny] {$1-\lambda_{\max} ( \Pi_{ \calS} \overline{\Pi}_{ \Sigma} \Pi_{ \calS} )$};

    \draw[dashed] (0, 0.8) -- (1, 0.8);
    \draw[dashed] (0, 0.75) -- (1, 0.75);
    \draw[dashed] (0, 0.2) -- (1, 0.2);
    \draw[dashed] (0, 0.25) -- (1, 0.25);
    \draw[dashed] (0.2, 0) -- (0.2, 1);
    \draw[dashed] (0.8, 0) -- (0.8, 1);
    \node[below] at (1, 0) {\scriptsize $1$};
    \node[left] at (0, 1) {\scriptsize $1$};
\end{tikzpicture}
\caption{Feasible region $\Theta$ of unit norm generalized graph signals $f$ with $\norm{\Pi_{\calS} f}_{2}=\alpha_{\calS}$ and $\norm{\Pi_{\Sigma} f}_{2}=\beta_{\Sigma}$}
\label{fig:feasible region}
\end{figure}
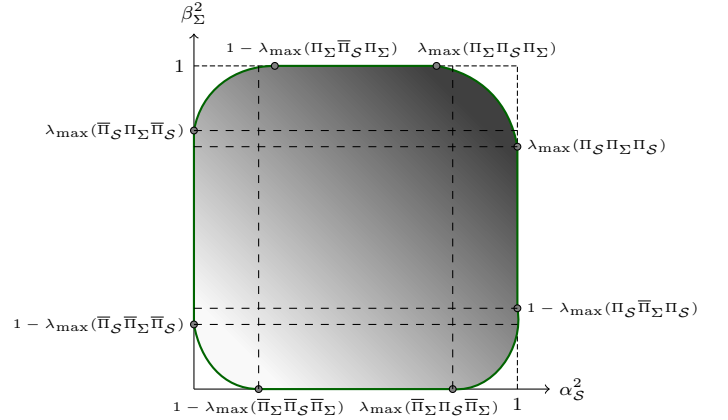

\Cref{fig:feasible region} illustrates the feasible region $\Theta$. The boundary curves at the four corners of $\Theta$ are derived from the equality conditions in \cref{eq.Theorem_feasible region}.
As $\calS$ expands for a given $\Sigma$, the upper boundary curve progressively approaches the upper right corner of $\Theta$. This convergence occurs when the projections $\Pi_{\calS}$ and $\Pi_{\Sigma}$ satisfy the conditions for perfect localization, as detailed in \cref{Sect.Localization properties}.

\section{Joint vertex-time dictionary for generalized graph signal reconstruction}
\label{sec:print}
Using the uncertainty principle, we construct a dictionary of signals that are maximally concentrated in the joint vertex-time and spectral-frequency domains. Fundamental atoms are identified as those signals with maximal concentration in these domains. We derive conditions for perfect localization in the dual domains of joint vertex-time and spectral-frequency and outline the optimization process for reconstructing signals using this dictionary.

\subsection{Localization properties}
\label{Sect.Localization properties}
We present the conditions under which a generalized graph signal can achieve perfect localization in the dual domains of vertex-time and spectral-frequency simultaneously in the following theorem.

\begin{Theorem}
\label{thm:localized_both_domains}
A nonzero generalized graph signal $f$ can be perfectly localized over both the joint vertex-time subset $\calS$ and spectral-frequency subset $\Sigma$ (i.e., $f\in \ima(\Pi_{\calS}) \cap \ima(\Pi_{\Sigma})$) if and only if $f$ is an eigenvector of the operator $\Pi_{\Sigma}\Pi_{\calS}\Pi_{\Sigma}$ (or equivalently, $\Pi_{\calS}\Pi_{\Sigma}\Pi_{\calS}$) with an eigenvalue of $1$, i.e.,
\begin{align}
\label{joint_localizd_conditions}
    \Pi_{\Sigma}\Pi_{\calS}\Pi_{\Sigma} f = f.
\end{align}
\begin{proof}
See \cref{prf:thm:localized_both_domains}.
\end{proof}
\end{Theorem}
\begin{Remark}

To give insights into the perfect localization conditions in \cref{thm:localized_both_domains}, we examine three special cases:
\begin{enumerate}[i)]
    \item  When $\calH = \Complex$, the signal $f \in L^{2}(\calV \times \calT)$ is a graph signal with $f(v) \in \Complex$ for each $v \in \calV$. If $f$ is an eigenvector of $\Pi_{\calW'}\Pi_{\calV'}\Pi_{\calW'}$ associated with a unit eigenvalue, then $f$ achieves perfect localization in both vertex and its dual domains, consistent with Theorem 2.1 of \cite{Tsitsvero2016}.
    \item  When the graph $\calG$ consists of a single vertex and $\calT = \Real$, $f\in L^{2}(\calV\times\calT)$ becomes a continuous-time function in the Hilbert space  $L^{2}(\Real)$. In this case, the operators $\Pi_{\Sigma}$ and $\Pi_{\calS}$ reduce to the frequency-limiting operator $\Pi_{\Omega'}$ and the time-limiting operator $\Pi_{\calT'}$, respectively. The condition in \cref{joint_localizd_conditions} fails to hold because
    the largest eigenvalue of $\Pi_{\Omega'}\Pi_{\calT'}\Pi_{\Omega'}$ is less than $1$ \cite{Pollak1961}, indicating that no non-zero signal can have both finite time duration and finite frequency bandwidth.
    \item Consider two subsets $\calS=\calV'\times \calT'$ and $\Sigma = \calW'\times \Omega'$. The operators $\Pi_{\calS}$ and $\Pi_{\Sigma}$ can be written as $\Pi_{\calS}=\Pi_{\calV'}\otimes\Pi_{\calT'}$ and $\Pi_{\Sigma}=\Pi_{\calW'}\otimes\Pi_{\Omega'}$, respectively. In this setup, a generalized graph signal $f$ cannot be perfectly localized over both $\calS$ and $\Sigma$ simultaneously, because the largest eigenvalue of $\Pi_{\Sigma}\Pi_{\calS}\Pi_{\Sigma}$, being the product of the largest eigenvalues of $\Pi_{\calW'}\Pi_{\calV'}\Pi_{\calW'}$ and $\Pi_{\Omega'}\Pi_{\calT'}\Pi_{\Omega'}$, is less than $1$. Thus, the perfect localization conditions in \cref{thm:localized_both_domains} cannot be satisfied.
\end{enumerate}
\end{Remark}
In cases where signals are not perfectly localized in both $\calS$ and $\Sigma$, it is crucial to identify signals with limited support in one domain and maximal concentration in the other. We say that a generalized graph signal is bandlimited if it is localized in the spectral-frequency domain, i.e., it is bandlimited in the usual sense in both the spectral and frequency domains.
Our goal is to find an orthonormal set of bandlimited signals, satisfying $\Pi_{\Sigma}f = f$, that are maximally concentrated in the joint vertex-time domain $\calS$. This is achieved by solving an optimization problem for a set of signals ${\xi_i}$, $i = 1, \dots, N$, iteratively as follows:
\begin{align}
\begin{aligned}\label{eq.optimization_problem}
\max_{\xi_{i}}\ & \norm{\Pi_{\calS} \xi_{i}}\\
\ST\ & \norm{\xi_{i}} = 1,\\
& \Pi_{\Sigma}\xi_{i}=\xi_{i},\\
& \angles{\xi_{i},\xi_{j}}=0, \forall j < i.
\end{aligned}
\end{align}
The solution to \cref{eq.optimization_problem} is provided in the following theorem.

\begin{Theorem}
\label{Theo.solu_optimization}
Let $\set{\xi_{i}}_{i=1,2,\dots}$ be the solution of \cref{eq.optimization_problem}, each maximally concentrated over the vertex-time set $\calS$. Then, these signals are the eigenvectors of the operator $\Pi_{\Sigma}\Pi_{\calS}\Pi_{\Sigma}$, i.e.,
\begin{align}
 \Pi_{\Sigma}\Pi_{\calS}\Pi_{\Sigma} \xi_{i} = \lambda_{i} \xi_{i}
\end{align}
where the eigenvalues $\lambda_{1}\geq\lambda_{2}\geq\cdot\cdot\cdot$ form a non-increasing sequence. Moreover, these signals are orthogonal over $\calS$, i.e., $\angles{\xi_{i},\Pi_{\calS}\xi_{j}} = \lambda_{j}\delta_{ij},$ where $\delta_{ij}$ is the Kronecker delta.
\begin{proof}
See \cref{prf:thm:solu_optimization}.
\end{proof}
\end{Theorem}

\begin{Remark}
\Cref{Theo.solu_optimization} offers a general framework for constructing a set of orthonormal bandlimited signals $\set{\xi_{i}}_{i\geq1}$, including both bandlimited signals in the graph domain and in the Hilbert space $L^2(\Real)$ as special cases.
\begin{enumerate} [i)]
    \item If $\calH=\Complex$, the set of orthonormal bandlimited signals corresponds to the set of orthonormal bandlimited vectors $\set{\phi_{i}}_{i=1,2,\dots,K}$, where $K=\rank(\Pi_{\calW'})$. These vectors are the eigenvectors of $\Pi_{\calW'}\Pi_{\calV'}\Pi_{\calW'}$, consistent with Theorem 2.3 in \cite{Tsitsvero2016}.

    \item When the graph $\calG$ consists of a single vertex, the orthonormal bandlimited signals are the \glspl{PSWF} $\set{\psi_{j}}_{j\geq1}$ \cite{Slepian1961,Laudau1962,Pollak1961}, derived from the eigenfunctions of $\Pi_{\Omega'}\Pi_{\calT'}\Pi_{\Omega'}$.

    \item When considering two subsets $\calS = \calV' \times \calT'$ and $\Sigma = \calW' \times \Omega'$, the vertex-time limiting operator $\Pi_{\calS}$ and the spectral-frequency limiting operator $\Pi_{\Sigma}$ are given by $\Pi_{\calS} = \Pi_{\calV'} \otimes \Pi_{\calT'}$ and $\Pi_{\Sigma} = \Pi_{\calW'} \otimes \Pi_{\Omega'}$, respectively. The orthonormal bandlimited signals are the tensor products of orthonormal bandlimited signals in the graph and Hilbert space $L^2(\Real)$, denoted as $\set{\xi_{i,j}}_{i,j\geq1} = \set{\phi_{i}\otimes\psi_{j}}_{i,j\geq1}$.
\end{enumerate}
\end{Remark}

\Cref{Theo.solu_optimization} provides a set of perfectly bandlimited signals with maximal concentration in the joint vertex-time domain. Next, we construct a dictionary of these signals to facilitate the reconstruction of bandlimited generalized graph signals.

\subsection{Joint vertex-time dictionary for generalized graph signal reconstruction}

An overcomplete dictionary with richer atoms enhances signal reconstruction by enabling sparse, flexible, and robust representations. Motivated by sensor networks with damaged or intermittently operational sensors, we use multiple subsets $\calS^{(i)}$, $i=1,2,\dots I$, to cover the joint vertex-time domain. For computational feasibility, we consider only subsets of the form $\calS^{(i)}=\calV'^{(i)}\times\calT'^{(i)}$, where $\calT'^{(i)}=[t_{c}^{(i)}-\ell\tc{i}/2,t_{c}^{(i)}+\ell\tc{i}/2]$. Our objective is to learn a dictionary bandlimited to $\Sigma=\calW'\times\Omega'$ and maximally concentrated on a collection of the subsets $\calS\tc{i}$.

Let $\Phi(\calV'^{(i)})=\left(\phi_{k}(\cdot;\calV'^{(i)})\right)_{k\geq1}$, where for each $k\geq 1$, $\phi_{k}(\cdot;\calV'^{(i)}):\calV\rightarrow \Real$ is a graph signal localized over $\calV'^{(i)}$, be the eigenvectors of $\Pi_{\calW'}\Pi_{\calV'^{(i)}}\Pi_{\calW'}$. Let $\Psi(\calT'^{(i)})=\left(\psi_{n}(\cdot;\calT'^{(i)})\right)_{n\geq1}$, where for each $n\geq1$, $\psi_{n}(\cdot;\calT'^{(i)}): \calT \rightarrow \Real$ is a PSWF localized over $\calT'^{(i)}$, be the eigenvectors of $\Pi_{\Omega'}\Pi_{\calT'^{(i)}}\Pi_{\Omega'}$. We construct a dictionary candidate as
\begin{align}
\Xi(\bt_{c}, \boldell) = \parens*{\Phi(\calV'^{(i)})\otimes\Psi(\calT'^{(i)})}_{i=1}^I,
\end{align}
where
\begin{align*}
\Phi(\calV'^{(i)})\otimes\Psi(\calT'^{(i)})=\parens*{\phi_{k}(\cdot;{\calV'}\tc{i}) \otimes \psi_{n}(\cdot;{\calT'}\tc{i})}_{k,n\geq1}
\end{align*}
is the set of signals localized over $\calS^{(i)}$,
$\bt_c=(t_c\tc{i})_{i\geq1}$,
and $\boldell=(\ell\tc{i})_{i\geq1}$.
While dictionaries with an analytical form and predefined parameters are effective in capturing the global structure of a signal and facilitating its reconstruction, learned dictionaries with optimized parameters often yield superior performance in specific applications, as they are better adapted to the observed class of signals \cite{Elad2006,Mairal2009,Rencker2019}.
We optimize the parameters $\bt_{c}$ and $\boldell$ using training samples to learn the dictionary $\Xi(\bt_{c}, \boldell)$ for signal reconstruction as follows.

Given noisy samples of a generalized graph signal $f$ at a subset $\calM \subset \calV \times \calT$ of vertices and time instances, the goal is to recover the signal $f$.
Suppose the sampling set is $\calM = \set{(v_{m},t_{m}) \given m=1,\dots,M}$ and the noisy observations are given by $y_{m}=f(v_{m},t_{m})+\epsilon_{m}$, for $m = 1,\dots,M$, where $\epsilon_{m}$ are \gls{iid} zero-mean noise with variance $\sigma^{2}$.
We assume that the signal of interest is localized in a finite interval $[0,\delta]$ in practice. Denote the training samples $\bY_{\mathrm{train}}=\left(y_{m}\right)_{m\in\{1,\dots,M\}}$, and the dictionary $\Xi(\bt_{c}, \boldell)$ evaluated at the training samples as $\Xi(\calM_{\mathrm{train}};\bt_{c}, \boldell)$, where $\calM_{\mathrm{train}}$ denotes train instances. The dictionary learning for signal recovery problem is formulated as:
\begin{align}
\begin{aligned}\label{eq.dictionary_optimize}
\argmin_{\bx, \bt_{c}, \boldell}\ & \norm{\bY_{\mathrm{train}}-\Xi(\calM_{\mathrm{train}}; \bt_{c}, \boldell)\bx}_{2}^{2}+\mu\norm{\bx}_{1}\\
\ST\ &  t_{c}^{(i)} \in [0,\delta],\ 0 \leq\ell\tc{i}\leq \delta,\ \forall i.
\end{aligned}
\end{align}
Here, $\bx$ contains sparse coefficients, and $\mu$ is a regularization parameter.

The optimization problem \cref{eq.dictionary_optimize} is not jointly convex but is convex with respect to each variable when the others are fixed. We solve the problem by alternating minimization over $\bx$, $\bt_{c}$ and $\boldell$ until convergence. At the $u$-th iteration, the following minimizations are performed sequentially:
\begin{align*}
\bx^{u+1}&\leftarrow\argmin_{\bx^{u}} \norm{\bY_{\mathrm{train}}-\Xi(\calM_{\mathrm{train}};\bt_{c}^{u},\boldell^{u})\bx^{u}}_{2}^{2}+\mu\norm{\bx^{u}}_{1} \\
\bt_{c}^{u+1}&\leftarrow\argmin_{\bt_{c}^{u}}\norm{\bY_{\mathrm{train}}-\Xi(\calM_{\mathrm{train}};\bt_{c}^{u},\boldell^{u})\bx^{u+1}}_{2}^{2},\\
&\qquad\qquad\ST t_{c}^{(i)} \in [0,\delta],\ \forall i,\\
\boldell^{u+1} &\leftarrow \argmin_{\boldell^{u}} \norm{\bY_{\mathrm{train}}-\Xi(\calM_{\mathrm{train}};\bt_{c}^{u+1},\boldell^{u})\bx^{u+1}}_{2}^{2},\\
&\qquad\qquad\ST 0 \leq \ell\tc{i}\leq \delta,\ \forall i.
\end{align*}
The sparse coefficients $\bx$ are updated using Lasso regression\cite{Tibshirani1996}. After updating $\bx$, the dictionary is updated by optimizing $\bt_{c}$ and $\boldell$ using the gradient descent algorithm \cite{Lecun1998}. We summarize the complete procedure in \cref{algo:mini}, referred to as JECD, where Loss denotes the objective function in \cref{eq.dictionary_optimize}. Suppose the optimization solution is given by $(\widetilde{\bx}, \widetilde{\bt}_{c}, \widetilde{\boldell})$. With the optimized dictionary $\Xi(\widetilde{\bt}_{c}, \widetilde{\boldell})$, the reconstruction process is formulated as:
\begin{align}
    \label{eq.reconstruction}
        \argmin_{\bx} &\norm{\bY_{\mathrm{test}}-\Xi(\calM_{\mathrm{test}}; \widetilde{\bt}_{c}, \widetilde{\boldell})\bx}_{2}^{2}+\mu\norm{\bx}_{1}
\end{align}
where $\calM_{\mathrm{test}}$ denotes test instances and $\bY_{\mathrm{test}}$ the corresponding test samples. Suppose the optimal solution is given by $\bx^*$. The reconstructed generalized graph signal is given by $\widehat{f}(v,t)=\Xi(\{(v,t)\}; \widetilde{\bt}_{c}, \widetilde{\boldell})\bx^*$ for each $(v,t)\in\calV\times\calT$. The recovery performance is measured by the relative square error:
\begin{align}
    \label{recovery_error}
    \mathrm{RSE} = \frac{\sum_{(v,t)\in \calM_\mathrm{val}}\left(f(v,t)-\widehat{f}(v,t)\right)^2}{\sum_{(v,t)\in \calM_\mathrm{val}} f(v,t)^2},
\end{align}
where $\calM_\mathrm{val}$ denotes a validation set of sample points sampled distinctly from $\calM_{\mathrm{test}}$.

\begin{algorithm}[!htb]
	\caption{Joint energy concentrated dictionary (JECD) learning algorithm}\label{algo:mini}
	\begin{algorithmic}[1]
            \State Input: Training instances $\calM_{\mathrm{train}}$, training samples $\bY_{\mathrm{train}}=\left(y_{m}\right)_{m\in\{1,\dots,M\}}$, learning rate $\eta_{1}$ and $\eta_{2}$, convergence tolerance $\epsilon$ and signal length $\delta$.
            \State Output: Optimized $\widetilde{\bt_{c}}$ and $\widetilde{\boldell}$
		\State Initialize $\bt_{c}^{0}$, $\boldell^{0}$, $u=0$.
		\Repeat
            \State $\bx^{u} = \mathrm{Lasso}(\bY_{\mathrm{train}},\Xi(\calM_{\mathrm{train}}; \bt_{c}^{u}, \boldell^{u}))$
            \State $\bt_{c}^{u+1}=\bt_{c}^{u}-\eta_{1}\ppfrac{}{\bt_{c}^{u}}\norm{\bY_{\mathrm{train}}-\Xi(\calM_{\mathrm{train}};\bt_{c}^{u}, \boldell^{u})\bx^{u}}_{2}^{2}$
            \State $\boldell^{u+1}=\boldell^{u}-\eta_{2}\ppfrac{}{\boldell^{u}}\norm{\bY_{\mathrm{train}}-\Xi(\calM_{\mathrm{train}};\bt_{0}^{u}, \boldell^{u})\bx^{u}}_{2}^{2}$
            \If {$\bt_{c}^{u+1}\notin [0,\delta] $ \text{or} $\boldell^{u+1}> \delta$}
            \State $\bt_{c}^{u+1}=\max(0,\min(\bt_{c}^{u+1}, \delta))$
            \State $\boldell^{u+1}=\max(0,\min(\boldell^{u+1}, \delta))$
            \EndIf
            \State $u\leftarrow u+1$
		\Until{$\norm{\mathrm{Loss}^{u}-\mathrm{Loss}^{u-1}} \leq \epsilon$}
	\end{algorithmic}
\end{algorithm}

\section{Experimental results}
We test JECD on a traffic dataset\footnote{The data is from the 3rd district of California and can be download from \url{https://pems.dot.ca.gov/}.} that captures high-resolution daily vehicle flow over two weekdays. The California Department of Transportation provides detailed traffic flow measurements (number of vehicles per unit interval) from stations along highways in Sacramento. Due to computational constraints, we focus on data from 70 stations over two weekdays (Monday and Tuesday) during the period of April 1-6, 2016. Following the approach in \cite{Loukas2019}, we use the road connectivity network from OpenStreetMap.org to construct a time series for each highway segment. The traffic flow on each segment is set as a weighted average of all nearby stations, with adjustments made according to traffic direction.

We randomly select a proportion $p_o$ of Monday's records as training samples to optimize the joint vertex-time dictionary $\Xi$. Using the optimized dictionary, we reconstruct Tuesday's records with a randomly selected proportion $p_o$ of Tuesday’s samples, evaluating reconstruction performance on the remaining $1-p_o$.
We compare JECD against the following joint vertex-time dictionaries for reconstructing generalized graph signals:
a) \textbf{\Gls{JFT}}: A joint vertex-time dictionary is constructed from a set of joint vertex-time Fourier bases $\set{\phi_{k}\otimes\psi_{l} \given k=1,\dots,K,\ l=1,\dots,L}$ \cite{Loukas2016}. Here, $\phi_{k}$ is the $k$-th orthonormal eigenvector of the graph Laplacian $\bL$, and $\psi_{l}$ is an element of the orthonormal Fourier basis in the time domain, given by $\set{\frac{1}{\sqrt{2\pi}}, \frac{1}{\sqrt{\pi}}\cos(lt), \frac{1}{\sqrt{\pi}}\sin(lt)}_{l\geq1}$. Both $K$ and $L$ are tunable hyperparameters.
b) \textbf{\Gls{STVFT}} \cite{Grassi2018}: We design the mother vertex-time kernel in the short time-vertex Fourier basis in a separable manner. In the vertex domain, following \cite{Grassi2018}, the localized graph basis is $\bh_{p,q} = (\bU h(\Lambda - \Lambda_{q}) \bU^{\ast}) \bdelta_{p}$, where $h(\cdot)$ is adopted as an Itersine kernel \cite{Perraudin2014}, given by $h(\Lambda) = \sin(0.5\pi\cos((\pi\Lambda)^2)$ with uniformly $Q$ translations $\Lambda_{q} = \diag(\frac{\lambda_{\max}}{Q}q)$ for $q = 1, \ldots, Q$. Here, $\bU$ is the unitary matrix whose columns are the Laplacian eigenvectors, $\Lambda$ is a diagonal matrix with non-negative real eigenvalues $\set{\lambda_{i}}_{i=1,\dots,N}$, and $\bdelta_p$ is a Kronecker delta centered at vertex $p$.
In the time domain, we use Gabor functions $g_{m,n}(t) = \frac{1}{\sqrt{2\pi}\rho}e^{-\frac{(t-m\tau_{0})^2}{2\rho^2}} e^{j n\omega_{0}}$ where $m, n \in \mathbb{Z}$, $\tau_{0}$ and $\omega_{0}$ are time and frequency shifts, and $\rho$ controls the Gaussian window width.
The joint vertex-time dictionary is constructed from the short time-vertex Fourier bases $\set{\bh_{p,q}\otimes g_{m,n}(t)}_{p,q,m,n}$. The parameters $Q$, $\tau_{0}$, $\omega_{0}$, and $\sigma$ are tunable hyperparameters.
c) \textbf{\Gls{STVWT}} \cite{Grassi2018}: The mother vertex-time kernel in the spectral time-vertex wavelet basis is also designed in a separable manner. For the vertex dimension, the localized graph wavelet basis is $\bh_{p,s}=(\bU h(s\Lambda)\bU^{\ast})\bdelta_{p}$, where $h(s\Lambda)$ is the scaled Itersine kernel \cite{Perraudin2014} and $\bdelta_p$ is a Kronecker delta centered at vertex $p \in \calV$.
For the time dimension, we use a Morlet wavelet $\psi_{a,b}(t) = \frac{1}{\sqrt{a}} e^{-\frac{(t-b)^2}{2a^2}} e^{j\omega_{0} \frac{t-b}{a}}$
with $a>0$ and $b\in \Real$ as the scale and translation parameters, controlling dilation and time shift respectively, and $\omega_{0}$ as the central frequency. The joint vertex-time dictionary is constructed from a collection of $\set{\bh_{p,s}\otimes\psi_{a,b}(t)}_{p,s,a,b}$. Here, $s$, $a$ and $b$ are tunable hyperparameters.

\cref{fig:Traffic_RMSE_Ratio} shows the reconstruction performance for different training ratios $p_{o}$, from $0.03$ to $0.15$. JECD consistently achieves the lowest RSE, demonstrating superior reconstruction accuracy compared to other methods. At lower training ratios (0.03 to 0.07), it shows a sharper RSE decline, indicating better learning efficiency from limited data. As the training ratio increases, JECD maintains a low and stable RSE, highlighting its robustness. To further demonstrate robustness, we evaluate reconstruction performance at varying Signal-to-Noise Ratio (SNR) levels with a fixed training ratio of 0.1. From \cref{fig:Traffic_RMSE_SNR}, we observe that JECD consistently yields the lowest RSE, showing superior noise resilience. Even at low SNRs, it maintains a lower RSE, and as SNR increases, it continues to outperform the others. JFT has the highest RSE across all SNR levels, indicating weaker performance in noisy conditions, while \gls{STVFT} and \gls{STVWT} perform better than JFT but still lag behind JECD, especially at lower SNRs. These results underscore JECD's effectiveness in handling noise with limited training data.

\begin{figure}[!htbp]
    \centering
    \includegraphics[width=0.9\columnwidth]{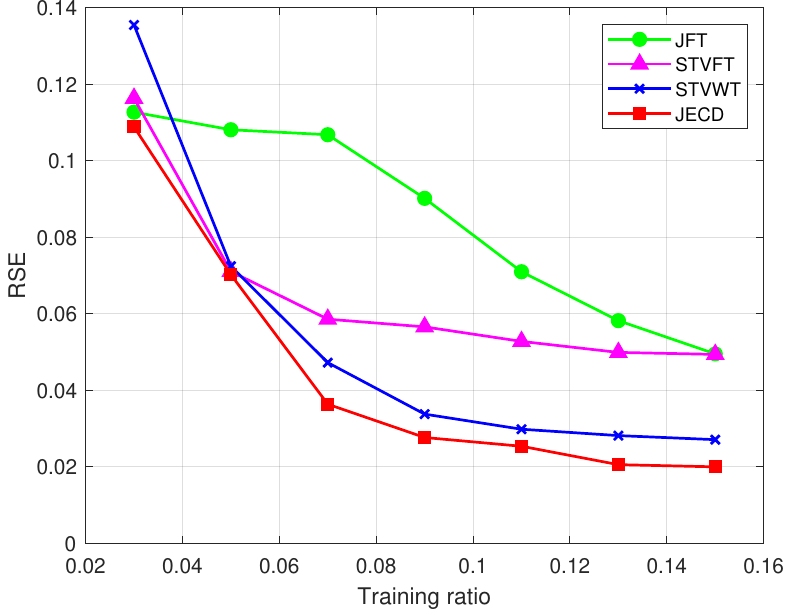}
    \caption{Reconstruction error over test set under different proportions of samples to be used for training. Each point in the figure is obtained by 10 repetitions.}
    \label{fig:Traffic_RMSE_Ratio}
\end{figure}
\begin{figure}[!htbp]
    \centering
    \includegraphics[width=0.9\columnwidth]{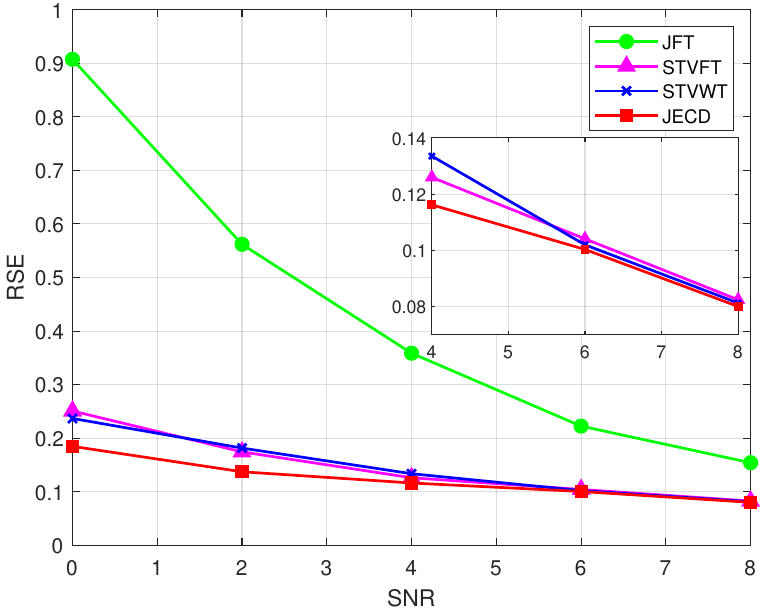}
    \caption{Reconstruction error over test set under different proportions of samples to be used for training. Each point in the figure is obtained by 10 repetitions.}
    \label{fig:Traffic_RMSE_SNR}
\end{figure}

\section{Conclusion}
\label{sec.conclusion}
We have derived an uncertainty principle for generalized graph signals, integrating the classical time-frequency and graph uncertainty principles into a unified framework.
We have introduced the joint vertex-time and spectral-frequency spreads to measure signal localization across these domains. Within this framework, we identified a class of signals with maximal concentration in both domains, which form the atoms of a new joint vertex-time dictionary. This dictionary proves highly effective for reconstructing signals under practical conditions when the data is incomplete or intermittent. Comparative experiments on a real-world dataset demonstrates that our method outperforms existing approaches, achieving superior reconstruction accuracy and enhanced noise robustness.

\appendices

\section{Proof of \cref{thm:UP_GGSP}}\label[Appendix]{prf:thm:UP_GGSP}
Before proceeding to the proof, we introduce several technical results that will help us establish \cref{thm:UP_GGSP}. The following definition and two lemmas are adapted from \cite{Tsitsvero2016} and \cite{Slepian1961}, which derive the minimum angle between two function spaces $\ima(\Pi_{\Sigma})$ and $\ima(\Pi_{\calS})$. The proof essentially follows the same procedure as in \cite{Pollak1961}, where it was originally stated for continuous-time signals.

\begin{Definition}
The angle between two functions $f$ and $g$ in a Hilbert space is given by
\begin{align}
\label{angle_defi}
\theta(f,g)=\cos^{-1}\frac{\Re\{\angles{f,g}\}}{\norm{f}\norm{g}}
\end{align}
where $\angles{}$ is the Hilbert space inner product, and $\norm{}$ its norm.
\end{Definition}

By the Cauchy-Schwarz inequality $\angles{f,g} \leq \norm{f}\norm{g}$ and it is clear that
\begin{align*}
    -1 \leq \frac{\Re\{\angles{f,g}\}}{\norm{f}\norm{g}} \leq 1
\end{align*}
and $\theta(f,g)=0$ if and only if $f = k \cdot g$ for some constant $k$, i.e., when the two functions are colinear.

Recall that $\ima(\Pi_{\Sigma})=\set{f \given \Pi_{\Sigma} f=f}$ and $\ima(\Pi_{\calS})=\set{g \given \Pi_{\calS} g=g}$. The following lemma provides the minimum angle between these two subspaces.

\begin{Lemma}
\label{Lemma_the_min}
For any function $f\in \ima(\Pi_{\Sigma})$,
\begin{align}
    \inf_{g\in \ima(\Pi_{\calS})} \theta(f,g) = \cos^{-1} \frac{\norm{\Pi_{\calS} f}}{\norm{f}}
\end{align}
is achieved by $g=k(\Pi_{\calS}f)$ for any constant $k>0$.
\begin{proof}
For any $g \in \ima(\Pi_{\calS})$ it holds that
\begin{align*}
\Re\{\angles{f,g}\}
& \leq \abs{\angles{f,g}} = \abs{\angles{f-\Pi_{\calS} f+\Pi_{\calS} f, g}}\\
&= \abs{\angles{\Pi_{\calS} f, g}} \leq \norm{\Pi_{\calS} f}\norm{g}
\end{align*}
where the second equality holds due to the fact that $f-\Pi_{\calS} f$ is orthogonal to $g$. Then the angle between these two functions $f$ and $g$ is given by
\begin{align}
\theta(f,g)
&= \cos^{-1} \frac{\Re\{\angles{f,g}\}}{\norm{f}\norm{g}} \\
& \geq \cos^{-1} \frac{\norm{\Pi_{\calS} f}\norm{g}}{\norm{f}\norm{g}}= \cos^{-1} \frac{\norm{\Pi_{\calS} f}}{\norm{f}}.
\end{align}
It follows that
\begin{align}
    \inf_{g \in \ima(\Pi_{\calS})} \theta(f,g) = \cos^{-1} \frac{\norm{\Pi_{\calS} f}}{\norm{f}}
\end{align}
with equality when $g$ and $\Pi_{\calS} f$ are proportional.
\end{proof}
\end{Lemma}

The spaces $\ima(\Pi_{\Sigma})$ and $\ima(\Pi_{\calS})$ form the minimum angle $\theta_{\min}$, also called the first principle angle \cite{Colub1973}, given by
\begin{align*}
    \theta_{\min}=\inf_{\substack{f \in \ima(\Pi_{\Sigma})\\ g\in \ima(\Pi_{\calS})}} \theta(f,g).
\end{align*}

\begin{Theorem}
\label{min_angle}
The minimum angle $\theta_{\min}$ between $\ima(\Pi_{\Sigma})$ and $\ima(\Pi_{\calS})$ is
\begin{align}
    \theta_{\min} = \cos^{-1}\sqrt{\lambda_{\max}\left(\Pi_{\Sigma}\Pi_{\calS}\Pi_{\Sigma}\right)}
\end{align}
and is achieved by $f=\psi_{0}$ and $g=\Pi_{\calS}\psi_{0}$, where $\psi_{0}$ is an eigenvector of the Hermitian operator $\Pi_{\Sigma}\Pi_{\calS}\Pi_{\Sigma}$ corresponding to the eigenvalue $\lambda_{\max}(\Pi_{\Sigma}\Pi_{\calS}\Pi_{\Sigma})$.
\begin{proof}
From \cref{Lemma_the_min}, we have
\begin{align}
    \inf_{\substack{f \in \ima(\Pi_{\Sigma})\\ g\in \ima(\Pi_{\calS})}} \theta(f,g)=\inf_{f\in \ima(\Pi_{\Sigma})}\cos^{-1}\frac{\norm{\Pi_{\calS} f}}{\norm{f}}
\end{align}
where the infimum on the left hand side is achieved if the infimum on the right hand side is achieved. Since $\cos(\cdot)$ decreases monotonically in $[0,\pi]$, the infimum of the right hand side is equivalent to the supremum of $\frac{\norm{\Pi_{\calS} f}}{\norm{f}}$ with $f \in \ima(\Pi_{\Sigma})$. For any $f\in \ima(\Pi_{\Sigma})$ satisfying $\Pi_{\Sigma} f=f$, we have
\begin{align}
\frac{\norm{\Pi_{\calS} f}_{2}^{2}}{\norm{f}_{2}^{2}}
&=\frac{\norm{\Pi_{\calS}\Pi_{\Sigma}f}_{2}^{2}}{\norm{f}_{2}^{2}}
= \frac{\langle \Pi_{\calS}\Pi_{\Sigma} f, \Pi_{\calS}\Pi_{\Sigma} f\rangle}{\norm{f}_{2}^{2}} \nn
& = \frac{\angles*{f, \Pi_{\Sigma}^{\ast}\Pi_{\calS}^{\ast}\Pi_{\calS}\Pi_{\Sigma} f}}{\norm{f}_{2}^{2}}.\label{pre_max}
\end{align}
Since $\Pi_{\Sigma}$ and $\Pi_{\calS}$ are both Hermitian and idempotent operators (i.e., $\Pi_{\Sigma}^{\ast}=\Pi_{\Sigma}$, $\Pi_{\calS}^{\ast}=\Pi_{\calS}$, $\Pi_{\Sigma}^{2}=\Pi_{\Sigma}$, and $\Pi_{\calS}^2=\Pi_{\calS}$), \cref{pre_max} yields
\begin{align}
\label{opti_Bf}
   \frac{\norm{\Pi_{\calS} f}_{2}^{2}}{\norm{f}_{2}^{2}}= \frac{\angles{f, \Pi_{\Sigma}\Pi_{\calS}\Pi_{\Sigma} f}}{\norm{f}_{2}^{2}}.
\end{align}
Using the Rayleigh-Ritz theorem, we obtain
\begin{align}
\label{min_angle_RR_theorem}
    \max_{f}\frac{\norm{\Pi_{\calS} f}_{2}^{2}}{\norm{f}_{2}^{2}} = \lambda_{\max}(\Pi_{\Sigma}\Pi_{\calS}\Pi_{\Sigma}),
\end{align}
and the maximum value is achieved by the eigenvector of the operator $\Pi_{\Sigma}\Pi_{\calS}\Pi_{\Sigma}$ corresponding to the eigenvalue $\lambda_{\max}(\Pi_{\Sigma}\Pi_{\calS}\Pi_{\Sigma})$.
\end{proof}
\end{Theorem}

We are now ready to prove \cref{thm:UP_GGSP}.
Without loss of generality, we derive which values of $\beta_{\Sigma}$ are attainable for every choice of $\alpha_{\calS}$, assuming the unit norm signal $f$. We consider two cases $\alpha_{\calS}=1$ and $\alpha_{\calS}\in (0,1)$ respectively. The case $\alpha_{\calS}=1$ means that all the energy of signal is supported only on $\calS$. According to \cref{min_angle_RR_theorem} and Theorem \cref{min_angle}, the supremum of $\beta^2_{\Sigma}$ is given by
\begin{align}
    \sup_{f\in \ima(\Pi_{\calS}), \norm{f}=1} \beta_{\Sigma}^2 = \lambda_{\max}(\Pi_{\calS}\Pi_{\Sigma}\Pi_{\calS})
\label{eq.sup_beta}
\end{align}
by exchanging the role of $\Pi_{\Sigma}$ and $\Pi_{\calS}$. Similarly, the infimum of $\beta^2_{\Sigma}$ is given by
\begin{align}
    \inf_{f\in \ima(\Pi_{\calS}), \norm{f}=1} \beta_{\Sigma}^2= 1-\lambda_{\max}(\Pi_{\calS}\overline{\Pi}_{\Sigma}\Pi_{\calS})
\label{eq.inf_beta}
\end{align}
which is equivalent to the case of maximal concentration on the complement subset $\bar\Sigma$. Thus, in the case of $\alpha_{\calS}=1$, $\beta_{\Sigma}^2$ lies in the interval $[1-\lambda_{\max}(\Pi_{\calS}\overline{\Pi}_{\Sigma}\Pi_{\calS}),\lambda_{\max}(\Pi_{\calS}\Pi_{\Sigma}\Pi_{\calS})]$.
Next we are going to consider the behavior of $\beta_{\Sigma}$ for $\alpha_{\calS}$ belong to $\alpha_{\calS} \in (0,1)$. We can decompose any signal $f$ as
\begin{align}
\label{sig_decom}
    f = \lambda \Pi_{\calS} f+ \gamma \Pi_{\Sigma}f+g,
\end{align}
where $g$ is a signal orthogonal to both $\ima(\Pi_{\calS})$ and $\ima(\Pi_{\Sigma})$. Our goal is to find the closest signal to $f$ in the space spanned by $\Pi_{\calS} f$ and $\Pi_{\Sigma} f$. To do this, we calculate the inner products of \cref{sig_decom} successively with $f$, $\Pi_{\calS}f$, $\Pi_{\Sigma} f$ and $g$, and arrive at the following system of equations:
\begin{align}
\label{in_pro_sys_equ}
\centering
\left\{\begin{aligned}
    &1 = \lambda \alpha_{\calS}^2 + \gamma \beta_{\Sigma}^2 + \angles{g,f},\\
    &\alpha_{\calS}^2 = \lambda \alpha_{\calS}^2 + \gamma \angles{\Pi_{\Sigma} f, \Pi_{\calS} f},\\
    &\beta_{\Sigma}^2 = \lambda \angles{\Pi_{\calS} f, \Pi_{\Sigma} f} + \gamma \beta_{\Sigma}^2,\\
    &\angles{f,g} = \angles{g,g}.
\end{aligned}\right.
\end{align}
After eliminating $\angles{g,f}$, $\lambda$ and $\gamma$ from \cref{in_pro_sys_equ}, we obtain
\begin{align}
\begin{aligned}
    \label{eliminating_results}
    &\beta_{\Sigma}^2-2\Re\{\angles{\Pi_{\calS} f,\Pi_{\Sigma} f}+\alpha_{\calS}^2 \\
    =&\left(1-\frac{\abs{\angles{\Pi_{\calS} f,\Pi_{\Sigma} f}}^2}{\alpha_{\calS}^2\beta_{\Sigma}^2}\right)-\norm{g}_{2}^{2}\left(1-\frac{\abs{\angles{\Pi_{\calS} f,\Pi_{\Sigma} f}}^2}{\alpha_{\calS}^2\beta_{\Sigma}^2}\right).
\end{aligned}
\end{align}
According to \cref{angle_defi}, we define
\begin{align}
\label{angle_betw_proof}
    \cos\theta = \frac{\Re\{\angles{\Pi_{\calS} f,\Pi_{\Sigma} f}\}}{\norm{\Pi_{\calS} f}\norm{\Pi_{\Sigma} f}}
\end{align}
to measure the angle $\theta$ between $\Pi_{\calS} f\in \ima(\Pi_{\calS})$ and  $\Pi_{\Sigma} f \in \ima(\Pi_{\Sigma})$. From \cref{min_angle}, the angle $\theta$ between $\Pi_{\calS} f\in \ima(\Pi_{\calS})$ and  $\Pi_{\Sigma} f \in \ima(\Pi_{\Sigma})$ would be larger than the minimum one, i.e.,
\begin{align}
\label{angle_bound}
    \theta \geq\cos^{-1}\sqrt{\lambda_{\max}\left(\Pi_{\Sigma}\Pi_{\calS}\Pi_{\Sigma}\right)}=\cos^{-1}\lambda_{\max}\left(\Pi_{\calS}\Pi_{\Sigma}\right).
\end{align}
Since $\norm{\Pi_{\calS} f}=\alpha_{\calS}$, $\norm{\Pi_{\Sigma} f}=\beta_{\Sigma}$, then \cref{angle_betw_proof} can be written as
\begin{align}
    \alpha_{\calS}\beta_{\Sigma} \cos\theta= \Re\{\angles{\Pi_{\calS} f,\Pi_{\Sigma} f}\} \leq \abs{\angles{\Pi_{\calS} f,\Pi_{\Sigma} f}} \leq \alpha_{\calS} \beta_{\Sigma},
\end{align}
and we can further write
\begin{align}
\label{angle_inequ}
    0 \leq 1- \frac{\abs{\angles{\Pi_{\calS} f,\Pi_{\Sigma} f}}^2}{\alpha_{\calS}^2\beta_{\Sigma}^2} \leq 1-\cos^{2}\theta.
\end{align}
Combining \cref{angle_betw_proof} and \cref{angle_inequ}, \cref{eliminating_results} can be written as
\begin{align}
\begin{aligned}
\label{eliminating_results_refined}
    &(\beta_{\Sigma} -\alpha_{\calS}\cos\theta)^2\\
    &=\beta_{\Sigma}^2-2\beta_{\Sigma}\alpha_{\calS}\cos\theta+\alpha_{\calS}^2\cos^2\theta\\
    &= \left( 1-\frac{\abs{\angles{\Pi_{\calS} f,\Pi_{\Sigma} f}}^2}{\alpha_{\calS}^2\beta_{\Sigma}^2}\right)-\norm{g}_{2}^{2}\left( 1-\frac{\abs{\angles{\Pi_{\calS} f,\Pi_{\Sigma}f}}^2}{\alpha_{\calS}^2\beta_{\Sigma}^2}\right)\\
    &\leq -\alpha_{\calS}^2\sin^2\theta+1-\cos^2\theta= (1-\alpha_{\calS}^2)\sin^2\theta,
\end{aligned}
\end{align}
where the equality holds if and only if $\angles{\Pi_{\calS} f, \Pi_{\Sigma} f}$ is real and $g=0$. Next, from \cref{eliminating_results_refined} with $\alpha_{\calS}= \cos(\cos^{-1}\alpha_{\calS})$, we have
\begin{align}
    \beta_{\Sigma} \leq \cos\left(\theta-\cos^{-1}\alpha_{\calS}\right),
\end{align}
from which it follows, by means of the bound \cref{angle_bound}, that
\begin{align}
\label{eq.beta_range}
    \beta_{\Sigma} \leq \cos\left(\cos^{-1}\sqrt{\lambda_{\max}\left(\Pi_{\Sigma}\Pi_{\calS}\Pi_{\Sigma}\right)}-\cos^{-1}\alpha_{\calS}\right).
\end{align}
Therefore,
\begin{align}
\label{beta_range2}
 \cos^{-1}\alpha_{\calS} + \cos^{-1}\beta_{\Sigma} \geq \cos^{-1}\sqrt{\lambda_{\max}(\Pi_{\Sigma}\Pi_{\calS}\Pi_{\Sigma})}.
\end{align}
where equality holds by letting
\begin{align}
    f = p\psi_{0} + q \Pi_{\calS}\psi_{0}
\label{eq.unique_f}
\end{align}
with
\begin{align}
\begin{aligned}
    p &= \sqrt{\frac{1-\alpha^2_{\calS}}{1-\lambda_{\max}(\Pi_{\Sigma}\Pi_{\calS}\Pi_{\Sigma})}}, \\
    q &=\frac{\alpha_{\calS}}{\sqrt{\lambda_{\max}(\Pi_{\Sigma}\Pi_{\calS}\Pi_{\Sigma})}}-\sqrt{\frac{1-\alpha_{\calS}^2}{1-\lambda_{\max}(\Pi_{\Sigma}\Pi_{\calS}\Pi_{\Sigma})}}
\label{eq.para_in_f}
\end{aligned}
\end{align}
and where $\psi_{0}$ is an eigenvector of the Hermitian operator $\Pi_{\Sigma}\Pi_{\calS}\Pi_{\Sigma}$ corresponding to the eigenvalue $\lambda_{\max}(\Pi_{\Sigma}\Pi_{\calS}\Pi_{\Sigma})$. Taking the fact $\norm{\overline{\Pi}_{\calS} f}= \sqrt{1-\alpha_{\calS}^2}$ and applying the same steps between \cref{sig_decom} and \cref{beta_range2} to the operators $\Pi_{\Sigma}\overline{\Pi}_{\calS}\Pi_{\Sigma}$, we obtain the inequality
\begin{align}
\begin{aligned}
\label{BD_B}
 \cos^{-1}\sqrt{1-\alpha_{\calS}^2} + &\cos^{-1}\beta_{\Sigma}\\
 &\geq \cos^{-1}\sqrt{\lambda_{\max}(\Pi_{\Sigma}\overline{\Pi}_{\calS}\Pi_{\Sigma})}.
\end{aligned}
\end{align}
Similarly, we apply the same steps between \cref{sig_decom} and \cref{beta_range2} to the operators $\overline{\Pi}_{\Sigma}\Pi_{\calS}\overline{\Pi}_{\Sigma}$ and $\overline{\Pi}_{\Sigma}\overline{\Pi}_{\calS}\overline{\Pi}_{\Sigma}$ by means of $\norm{\overline{\Pi}_{\Sigma} f}=\sqrt{1-\beta^2_{\Sigma}}$, to further obtain
\begin{align}
\begin{aligned}
     \cos^{-1}\alpha_{\calS} + \cos^{-1}&\sqrt{1-\beta_{\Sigma}^2} \\
     &\geq \cos^{-1}\sqrt{\lambda_{\max}(\overline{\Pi}_{\Sigma}\Pi_{\calS}\overline{\Pi}_{\Sigma})}
\end{aligned}
\end{align}
and
\begin{align}
\begin{aligned}
     \cos^{-1}\sqrt{1-\alpha_{\calS}^2} &+ \cos^{-1}\sqrt{1-\beta_{\Sigma}^2}\\
     &\geq \cos^{-1}\sqrt{\lambda_{\max}(\overline{\Pi}_{\Sigma}\overline{\Pi}_{\calS}\overline{\Pi}_{\Sigma})}.
\end{aligned}
\end{align}

We have completed the proof of the four inequalities in \cref{eq.Theorem_feasible region}. For $\beta_{\Sigma}=1$, $\alpha_{\calS}^2$ lies in the interval $[1-\lambda_{\max}(\Pi_{\Sigma}\overline{\Pi}_{\calS}\Pi_{\Sigma}),\lambda_{\max}(\Pi_{\Sigma}\Pi_{\calS}\Pi_{\Sigma})]$,
and the concentrations are achievable by the eigenvectors of $\Pi_{\Sigma}\Pi_{\calS}\Pi_{\Sigma}$ which belong to $\Pi_{\Sigma}$ and their linear combinations. Continuing
by analogy one can show that all the values $\alpha_{\calS}$ and $\beta_{\Sigma}$ belonging to the border of $\Theta$ are achievable.  All the points inside $\Theta$ are achievable by the functions build up from different combinations of left and right singular vectors of $\Pi_{\Sigma}\Pi_{\calS}\Pi_{\Sigma}$,  $\Pi_{\Sigma}\overline{\Pi}_{\calS}\Pi_{\Sigma}$, $\overline{\Pi}_{\Sigma}\Pi_{\calS}\overline{\Pi}_{\Sigma}$ and $\overline{\Pi}_{\Sigma}\overline{\Pi}_{\calS}\overline{\Pi}_{\Sigma}$. This concludes the proof of \cref{thm:UP_GGSP}.

\section{Proof of \cref{thm:localized_both_domains}}\label[Appendix]{prf:thm:localized_both_domains}
To prove the forward direction, we repeatedly apply \cref{eq.vertex_time_limited_signals} and \cref{eq.spec_fre_limited_signals} to obtain
\begin{align}
    \Pi_{\Sigma}\Pi_{\calS}\Pi_{\Sigma} f = \Pi_{\Sigma}\Pi_{\calS} f = \Pi_{\calS}\Pi_{\Sigma} f = \Pi_{\Sigma} f = f,
    \label{BDBf_Bf_f}
\end{align}
which shows that $f$ is an eigenvector of $\Pi_{\Sigma}\Pi_{\calS}\Pi_{\Sigma}$ with a unit eigenvalue if it is perfectly localized in both domains. Conversely, if $f$ is an eigenvector of $\Pi_{\Sigma}\Pi_{\calS}\Pi_{\Sigma}$ with a unit eigenvalue, then
\begin{align}
    \Pi_{\Sigma}\Pi_{\calS}\Pi_{\Sigma} f = f.
    \label{eq.eigenvector_BDB}
\end{align}
Multiplying \cref{eq.eigenvector_BDB} by $\Pi_{\Sigma}$ and using $\Pi_{\Sigma}^2=\Pi_{\Sigma}$, we have
\begin{align}
    \Pi_{\Sigma}\Pi_{\calS}\Pi_{\Sigma} f = \Pi_{\Sigma} f.
    \label{eq.eigenvector_BDB_Bf}
\end{align}
By equating \cref{eq.eigenvector_BDB} and \cref{eq.eigenvector_BDB_Bf}, we obtain
\begin{align}
    \Pi_{\Sigma} f =f,
    \label{eq.spe_fre_localized}
\end{align}
indicating perfect localization in the joint spectral-frequency domain. Utilizing the Rayleigh-Ritz theorem \cite{Horn1985}, along with \cref{eq.spe_fre_localized} and the Hermitian property of $\Pi_{\Sigma}$ (i.e., $\Pi_{\Sigma}^{\ast}=\Pi_{\Sigma}$), we have
\begin{align}
   \max_{f} \frac{\angles{f,\Pi_{\Sigma}\Pi_{\calS}\Pi_{\Sigma}f}}{\angles{f,f}} = \max_{f}\frac{\angles{f,\Pi_{\calS} f}}{\angles{f,f}}=1,
\end{align}
which confirms that $f$ also satisfies $\Pi_{\calS} f = f$, meaning that it is perfectly localized in the joint vertex-time domain. This completes the proof of \cref{thm:localized_both_domains}.

\section{Proof of \cref{Theo.solu_optimization}}\label[Appendix]{prf:thm:solu_optimization}

Substituting the joint band-limiting constraint into the objective function in \cref{eq.optimization_problem}, we obtain
\begin{align}
    \begin{aligned}
        \xi_{i} = &\argmax_{\xi_{i}}~\norm{\Pi_{S}\Pi_{\Sigma}\xi_{i}}\\
        & \ST \norm{\xi_{i}} = 1, ~\angles{\xi_{i},\xi_{j}}=0,~j\neq i.
    \end{aligned}
\label{eq.optimization_problem_v2}
\end{align}
By the Rayleigh-Ritz theorem \cite{Horn1985}, the solutions are the eigenvectors of $(\Pi_{S}\Pi_{\Sigma})^{\ast}\Pi_{S}\Pi_{\Sigma}=\Pi_{\Sigma}\Pi_{S}\Pi_{\Sigma}$, i.e., $\Pi_{\Sigma}\Pi_{S}\Pi_{\Sigma} \xi_{i} = \lambda_{i} \xi_{i}$. Since $\Pi_{\Sigma}\xi_{i}=\xi_{i}$ and $\left(\Pi_{\Sigma}\right)^{\ast}=\Pi_{\Sigma}$, it follows that
\begin{align}
    \angles{\psi_{i},\Pi_{\Sigma}\Pi_{S}\Pi_{\Sigma}\psi_{j}}=\lambda_{j}\delta_{ij}.
\end{align}
This completes the proof of \cref{Theo.solu_optimization}.





\bibliographystyle{IEEEtran}
\bibliography{refs}

\end{document}